\documentclass[10pt]{article}
\usepackage[margin=1in]{geometry}
\usepackage{bm}
\usepackage{setspace}
\usepackage{graphicx}
\usepackage{amssymb}
\usepackage{amsmath}
\usepackage{epstopdf}

\newtheorem{defi}{Definition}[section]

\newtheorem{prop}{Proposition}[section]
\newtheorem{ex}{Example}[section]

\newtheorem{proof}{Proof}
\begin{document}
\title{From Boundary Crossing of Non-Random Functions to Boundary Crossing
of Stochastic Processes}
\begin{center}
\large \textbf{From Boundary Crossing of Non-Random Functions to Boundary Crossing
of Stochastic Processes}\normalsize\\~\\
MARK BROWN\\
Department of Mathematics, The City College of New York\\
New York, NY 10031-9100, USA\\
Email: cybergarf@aol.com\\~\\
VICTOR H. DE LA PE\~{N}A\\
Department of Statistics, Columbia University\\
New York, NY 10027-6902, USA\\
Email: vp@stat.columbia.edu\\~\\
TONY SIT\\
Department of Statistics, Columbia University\\
New York, NY 10027-6902, USA.\\
Email: tony@stat.columbia.edu\\~\\
\end{center}\large
\begin{abstract} 
One problem of wide interest involves estimating expected crossing-times. Several tools have been developed to solve this problem beginning with the works of Wald and the theory of sequential analysis. An extension of his approach is provided by the optional sampling theorem in conjunction with martingale inequalities. Deriving the explicit close form solution for the expected crossing times may be difficult. In this paper, we provide a framework that can be used to estimate expected crossing times of arbitrary stochastic processes. Our key assumption is the knowledge of the average behavior of the supremum of the process. Our results include a universal sharp lower bound on the expected crossing times.\\

Let $\{X_t\}_{t\geq 0}$ be a non-negative, measurable process. Set $\varsigma_t = \sup_{0 \leq s \leq t}X_s, \newline t>0$ and $T_r = \inf\{t > 0: X_t \geq r\}, r>0$. In this paper, we present bounds on the expected hitting time of $X_t$ via a nonrandom function $a(t) := E[\varsigma(t)]$ with $a^{-1}(s) = \infty$ if $a(\infty) < s$. In particular, we derive the following sharp universal lower bound , $E[a(T_r)] \geq r/2$  and if in addition $a(t)$ is concave, we have $E[T_r]\geq a^{-1}(r/2)$. This bound has the optimality property that in the case of non-random continuous processes, the following identity is valid: $T_r = a^{-1}(r)$. Therefore, (up to the constants) the method provides a unified approach to boundary crossing of non-random functions and stochastic processes. Furthermore, for a wide class of time-homogenuous, Markov processes, including Bessel processes, we are able to derive an upper bound $E[a(T_r)] \leq 2r$, which implies that $\sup_{r > 0}\bigg|\frac{E[a(T_r)]-r}{r}\bigg| \leq 1$. This inequality motivates our claim that $a(t)$ can be viewed as a natural clock for all such processes. The cases of multidimensional processes, non-symmetric and random boundaries are handled as well. We also present applications of these bounds on renewal processes in Example 10 and other stochastic processes.
\end{abstract}\normalsize
\section{Introduction}
\label{intro}
One problem of wide interest in the study of stochastic processes involves estimating $E[T_r]$, the expected time at which a process crosses a boundary. Following the work of \cite{delapena97}, we let $X_t$, $t \geq 0$, be a measurable process with the first passage time $T_r = \inf\{t > 0 : X_t \geq r\}$, $r>0$ at the level $r$. Deriving the explicit closed form solution for $E[T_r]$ can sometimes be difficult and so we are interested in finding the possibility of obtaining bounds for $E[T_r]$ above or below by some functions that are related to $a(t) := E[\sup_{0 \leq s\leq t}X_s]$. This approach was introduced in de la Pe\~na \cite{delapena97} (published in Section 2.7 of \cite{delapenagine}) in which the author constructed bounds for $E[T_r]$ for a class of Banach-valued processes with independent increments via decoupling. The bounds derived are of interest in applications where the moments of the maximal process can be readily obtained.\\

The main idea consists of an extension of the concept of boundary crossing by non-random functions to the case of random processes. $a(t)$ can be intuitively interpreted as a natural clock for all processes with the same $a(t)$. Here, we assume that we have information on $a(t)$. Furthermore, we assume that $X_t$ is a measurable separable process, $\sup_{0 \leq s \leq t}X_s$ is, therefore, well-defined and $T_r$ is measurable.\\

We would like to draw an attention to readers that the results derived in this paper provide a decoupling reinterpretation (in the context of boundary crossing) of the results due to Wald that concern randomly stopped sums of independent random variables (see \cite{wald45}), to Burkh\"older and Gundy on randomly stopped processes with independent increments (see \cite{burkholdergundy70}) and to Klass on bounds for randomly stopped Banach space-valued random sums (see \cite{klass88} and \cite{klass90}).\\

Recall that for $\{X_i\}$ a sequence of iid random variables adapted to $\mathcal{F}_i = \sigma(X_1, \ldots, X_i)$ with $E[X_1] = \mu$, where $|\mu| < \infty$, we have
\begin{equation*}
E[S_T] = \mu E[T], 
\end{equation*} 
where $T$ is a stopping time adapted to $\sigma\{X_i\}$, $S_n = \sum^n_{i=1}X_i$ and $E[T] < \infty$. Furthermore, if $E[X_1] = 0$ and $E[X^2_1] < \infty$, then
\begin{equation*}
E[S^2_T] = E[X^2_1]E[T], 
\end{equation*}
whenever $E[T] < \infty$. These celebrated results are known as Wald's first and second equations. As we can observe, $S_T = \sum^\infty_{i=1} X_i I(T \geq i)$ is composed of summands $X_i I(T \geq i)$ that are products of independent variables. This motivate the idea of replacing $\{X_iI(T\geq i)\}$ by $\{\widetilde{X}_iI(T \geq i)\}$, where $\{\widetilde{X}_i\}$ is an independent copy of $\{X_i\}$ and independent of $T$ as well. As a result, the above Wald's equations can be viewed from the decopuling perspective where, if we denote $\widetilde{S}_n = \sum^n_{i=1} \widetilde{X}_i$, then, whenever $E[T] < \infty$ and $E[X^2_1] < \infty$,
\begin{equation*}
E[S_T] = E[\widetilde{S}_T] = E[S_{\widetilde{T}}] \quad\quad\text{and}\quad\quad E[S^2_T] = E[\widetilde{S}^2_T] = E[S^2_{\widetilde{T}}].
\end{equation*}
It is important to realize that the variables $S_T$ and $\widetilde{S}_T$ can have drastically different behavior. One may consider a case in which $X_i$'s are iid random variables with $\Pr\{X_i=1\} = \Pr\{X_i = -1\} = \frac{1}{2}$. Define $T = \inf\{n > 0 : S_n = a \text{~or~} S_n = -b\}$ for some integers $a, b > 0$. It is easy to see that $S_T$ can only take a value of either $a$ or $-b$ whereas $\widetilde{S}_T$'s value is not restricted to these two choices.\\

An extension of Wald's second equation to the case of independent random variables having finite second moments with $T$ as a stopping time defined on the $X$'s was studied by de la Pe\~na and Govindarajulu \cite{delapenagov}. The following bound is sharp: 
\begin{equation}
E[S^2_T] \leq 2E[\widetilde{S}^2_T],
\end{equation}
where $S_n = \sum^n_{i=1}X_i$ and with $\widetilde{S}$ an independent copy of $S$. If we let $T_r = \inf\{n: S^2_n \geq r\}$ and $a(n) = E[\max_{1 \leq j \leq n}S^2_j]$, then
$$
r \leq E[S^2_{T_r}] \leq 2 E[\widetilde{S}^2_{T_r}] \leq 2 E\left[\max_{1 \leq j \leq T_r}\widetilde{S}^2_j\right],
$$
hence
\begin{equation*}
E[a(T_r)] \geq \frac{r}{2}.
\end{equation*}
which is closely related to our main result; see Proposition 1 below.\\

Along this vain, Klass \cite{klass88} obtained results on a best possible improvement of Wald's equation. In his work, Klass obtained bounds for stopped partial sum of independent random elements taking values if Banach space $(\mathcal{B}, \|\cdot\|)$. To be specific, he derived that 
\begin{equation*}
E\left[\max_{1 \leq n \leq T}\Phi(\|S_n\|)\right] \leq 20(18^\alpha)E\left[\max_{1 \leq n \leq T}\Phi(\|\widetilde{S}_n\|)\right],
\end{equation*}
where $\Phi(\cdot): [0, \infty)$ is a nondecreasing function such that $\Phi(0) = 0$ and for $\alpha >0$, $\Phi(cx) \leq c^\alpha\Phi(x)$ for all $c>2, x>0$ . Furthermore, in \cite{klass90}, the corresponding lower bound was also derived and hence
\begin{equation}
c_\alpha E\left[\max_{1 \leq n \leq T}\Phi(\|\widetilde{S}_n\|)\right] \leq E\left[\max_{1\leq n\leq T}\Phi(\|S_n\|)\right] \leq 20(18^\alpha) E\left[\max_{1 \leq n \leq T}\Phi(\|\widetilde{S}_n\|)\right].\label{eq:Klass}
\end{equation}
The corresponding counterpart for processes defined in continuous time domain was developed by \cite{delapenaeisenbaum94} (see also \cite{delapena96}) in which they showed
\begin{equation*}
c_\alpha E\left[\sup_{1 \leq t \leq T}\Phi(\|\widetilde{N}_t\|)\right] \leq E\left[\sup_{1\leq t\leq T}\Phi(\|N_t\|)\right] \leq 20(18^\alpha) E\left[\sup_{1 \leq t \leq T}\Phi(\|\widetilde{N}_t\|)\right],
\end{equation*}
where $N_t, t\geq 0$ is a $\mathcal{B}$-valued process continuous on the right with limits from the left with independent increments with $(\mathcal{B}, \|\cdot\|)$ is a separable Banach space.
As observed above, the inclusion of $a(t)$ facilitates the decopuling between the random stopping time $T$ and the underlying process $X_t$. All these results are closely connected with sequential analysis, for details, readers are referred to Lai \cite{lai01} for a recent survey.\\

In our results, below, we use $E[a(T_r)]$ as a key quantity. Here, we review our interpretation of this quantity. If we define 
\begin{equation*}
M(t) = \sup_{0 \leq s \leq t} X(s),\text{and } X_t \text{~is~ continuous}
\end{equation*}
then, $M(T_r) \equiv r$ with probability one. But, in contrast,
\begin{equation*}
E[a(T_r)] := \int (E[M(t)])dF_{T_r}(t) = \int (E[\widetilde{M}(t)])dF_{T_r}(t),
\end{equation*}
which can be interpreted as $E[a(T_r)] = E[a(\widetilde{T}_r)] = E[M(\widetilde{T}_r)]$, where $\widetilde{T}_r =_d T_r$ with $T_r$ and $\widetilde{T}_r$ independent. Thus it fits into the decoupling theme discussed above. This is further discussed in the remarks following Proposition 1.\\

In this paper, we devlope upon \cite{brownetal11} and \cite{delapenaetal11} to provide a universal sharp bound for $E[a(T_r)] \geq \frac{r}{2}$, which under a concavity assumption on $a(t)$ gives $E[T_r] \geq a^{-1}(\frac{r}{2})$; compare de la Pena and Yang \cite{delapenayang} in which the following bound is presented:
\begin{equation}
E[T_r] \geq (1-\epsilon)a^{-1}(\epsilon r), \quad \epsilon \in (0,1).
\end{equation} In addition, for a wide class of processes, we show that $$a^{-1}(\frac{r}{2}) \leq E[T_r] \leq a^{-1}(2r)$$ as well as the stability property $\displaystyle\sup_{r > 0}\bigg|\frac{E[a(T_r)]-r}{r}\bigg| \leq 1$. The above result, coupled with Eq. (1.1) of \cite{delapenayang}, without the concavitiy assumption on $a$, gives:
\begin{equation*}
\frac{1}{2}a^{-1}(\frac{r}{2}) \leq E[T_r] \leq a^{-1}(2r).
\end{equation*}


It should be noted that, the theory of first passage times for random processes has been extensively developed in recent times. In particular, the distribution of the first hitting times of Brownian motion has been studied through inverse Gaussian distribution; see \cite{sehadri94}. A similar approach to the first hitting times involving Le\'vy processes is also available. Readers may refer to \cite{sato99} for details.  

The typical methods, however, assume full knowledge of the distribution. In contrast, our approach provides bounds for all processes with a common $a(t) = E[\sup_{0 \leq s \leq t}X_s]$ based on the approximate knowledge of moments of the maximal process, or $E[X_t]$. Even in situations when the distribution of the process is known, the quantity $E[T_r]$ might not be easily obtained as shown in Example 7 in which we study the relative growth of the boundary crossing of a three-dimensional Brownian motion and related processes. (Renewal processes)\\ 

The rest of the paper is organized as follows: In section 2, we obtain upper and lower bounds on $E[T_r]$, as well as bounds on $E[a(T_r)]$. Section 3 elaborates some possible extensions of our methodology that can handle siutations in which the expected first hitting time is hard to obtain. An application of the bounds derived are presented in Section 4, followed by Section 5, which summarizes the paper. 

\section{Main results}
With the above definitions, let $a^{-1}(\xi) = \inf\{t > 0 : a(t) \geq \xi\}$ and $a^{-1}(\xi) = \infty$ if $a(t) < \xi, \forall t$. We have the following proposition. 
\begin{prop}
\label{prop:lowerboundaTr}
For all non-negative, measurable process $X_t$ with \newline $E\left[\sup_{0 \leq s \leq t}X_s \right] = a(t)$,
\begin{equation}
\frac{r}{2} \leq E[a(T_r)]
\label{eq:eatrlower}
\end{equation}
and the bound is sharp. Furthermore, if $a(t)$ is assumed to be concave, we obtain
\begin{equation}
a^{-1}(\frac{r}{2}) \leq E[T_r].
\label{eq:etrlower}
\end{equation}
\end{prop}
A decoupling reinterpretation of \eqref{eq:eatrlower} is given as follows:
Assuming that $X_t$ is continuous, we have
$$r = E\left[\sup_{0 \leq s \leq T_r}X_s\right] = E[M(T_r)],$$ we can rewrite \eqref{eq:eatrlower} as
\begin{equation*}
\frac{r}{2} = \frac{1}{2}E\left[\sup_{0 \leq s \leq T_r}X_s\right] \leq E\left[\sup_{0 \leq s \leq T_r}\widetilde{X}_s\right],
\end{equation*}
in which case the underlying process and its random stopping time are independent (decopuled).
\begin{proof}\label{pf:one}
Let $F$ be the cdf of $T_r$. If $F$ is continuous, then
\begin{align*}
a(t) & =  E\left[\sup_{0 \leq s \leq t} X(s)\right]\\
     & \geq  r \Pr\left\{\sup_{0 \leq s \leq t} X(s) \geq r\right\}\\
     & =  r\Pr\{T_r \leq t\} = rF(t).
\end{align*}
Due to the continuity of $F$, it is easy to see that $F(T_r) \sim \mathcal{U}(0,1)$ and $E[F_{T_r}(T_r)] = \frac{1}{2}$. As a result, we can conclude that $E[a(T_r)] \geq \frac{r}{2}$. More generally if the distribution of $T_r$ is not necessarily continuous then since
\begin{align*}
E[F(T_r)] & =  \Pr\left(\widetilde{T}_r \leq T_r\right)\\
          & =  \Pr\left(\widetilde{T}_r = T_r \right) + \frac{1}{2}\Pr\left(\widetilde{T}_r \neq T_r\right)\\
          & =  \frac{1}{2}\left[1 + \Pr\left(T_r = \widetilde{T}_r\right)\right] \geq \frac{1}{2},
\end{align*}
it follows that $E[a(T_r)] = rE[F_{T_r}(T_r)] \geq r/2$. To prove the sharpness of the bound, consider
\begin{equation*}
X_t = rI(t \geq U),
\end{equation*}
where $U \sim \mathcal{U}(0,1)$, then $T_r =_d U$ and $E[a(T_r)] = E[rU] = \frac{r}{2}$. To prove \eqref{eq:etrlower}, it follows immediate by Jensen's inequality that
\begin{align*}
-a(E[T_r]) & \leq -E[a(T_r)] \leq -\frac{r}{2}\\
\Longrightarrow~~~~~~~~~ E[T_r] & \geq  a^{-1}\left(\frac{r}{2}\right).
\end{align*}
\end{proof}

\noindent One may be interested in obtaining an upper bound for $E[a(T_r)]$ or $E[T_r]$ without further assumption. This is, however, impossible. The reason is that, without assumptions that control the growth of the process $X_t$ for any $t > T_r$, the value of $a(t)$ can blow up. Below we introduce a counter-example that demonstrates the impossibility that an upper bound can be obtained without further assumption.\\

\begin{ex} 
\label{eg:exp}
Let $X_t = tY,$ with $Y$ a non-negative random variable. Then 
\begin{equation*}a(t) = tE[Y]\end{equation*} and 
\begin{equation*}E[T_r] = rE[Y^{-1}].\end{equation*}

Suppose $Y$ is exponentially distributed with mean 1, $E[Y]=1$ while $E[Y^{-1}] = \infty$, so the behavior of $a(t)$ is controlled; however, $E[a(T_r)] = rE[Y]E[Y^{-1}] = rE[Y^{-1}]$ can be arbitrarily large. Controlling the growth of the proces after $T_r$ is reached enables us to derive an upper bound as shown in Proposition \ref{prop:upperboundaTr}. 
\end{ex}
\begin{defi}
A real random variable $A$ is less than a random variable $B$ in the ``usual stochastic order'' if
\begin{equation*}
\Pr(A>x) \le \Pr(B>x)\text{ for all }x \in (-\infty,\infty),
\end{equation*}
which is denoted $A \le_{st} B$; see \cite{rogerwilliams} and \cite{ross96}.
\end{defi}
\begin{prop}
\label{prop:upperboundaTr}
Assume that $X_t$ is non-negative and continuous, and in addition $X_t$ is a time homogeneous Markov process and that $T_{kr} - T_{(k-1)r} \geq_{st} T_r$, for $k \geq 2$, then
\begin{equation}
\label{eq:upperboundEaTr}
E[a(T_r)] \leq 2r 
\end{equation}
and
\begin{equation}
E[T_r] \leq a^{-1}(2r).
\label{eq:etrupper}
\end{equation}
\end{prop}
\begin{proof}
First notice that if $X_t$ is a Markov process with continuous paths - irreducible state space $[0,A]$, $0 < A \leq \infty$. 
If $r<A$, then $Pr\{T_r < \infty\} = 1$. Observe that, because 
of the continuous paths, we have 
\begin{equation*}
T_r = T_s + (T_r - T_s), s<r,
\end{equation*}
thus $T_r$ is stochastically greater than $T_s$ and $T_r-T_s$. Now
\begin{equation*}
Pr\{T_r > x+y\} = Pr\{T_r > x\} Pr\{T_r - x > y | T_r > x\}
\end{equation*}
and 
\begin{equation*}
T_r - x | T_r > x, X(x) = y \sim T_r - T_y,
\end{equation*}
by the Markov property, and since $T_r - T_z \leq_{st} T_r$, for $0 \leq z < t$, $Pr\{T_r - x > y | T_r > x\} \leq Pr\{T_r > y\},$ so that, if we define $\bar{F}(\cdot) = 1-F(\cdot)$,
\begin{equation}
\label{eq:NBU}
\bar{F}_{T_r}(x+y) \leq \bar{F}_{T_r}(x)\bar{F}_{T_r}(y),
\end{equation}
which is the submultiplicative, or new better than used (NBU), property. For details about NBU property, see \cite{barlowproschan75} and \cite{brown06}. Since $T_r$ is NBU, it has a finite mean. \\

Define $\bar{G}(t) := \frac{1}{\mu}\int^\infty_{t}\bar{F}_{T_r}(x)dx$, where $\mu = E[T_r]$, the stationary renewal distribution corresponding to $T_r$, since
\begin{align*}
\bar{G}(t) & =   \frac{\bar{F}_{T_r}(t)}{\mu}\int^\infty_t \frac{\bar{F}_{T_r}(x)}{\bar{F}_{T_r}(t)}dx\\
           &\leq \frac{\bar{F}_{T_r}(t)}{\mu}\int^\infty_t \bar{F}_{T_r}(x-t)dx\\
           & =  \frac{\bar{F}_{T_r}(t)}{\mu}\mu = \bar{F}_{T_r}(t),
\end{align*}
it follows that $G \leq_{st} F_{T_r}$.\\

The stationary renewal distribution corresponding to $F_{T_r}$ has $X^*_1 \sim G$, and $\{X_i\}_{i\geq 1} \sim F_{T_r}$. 
It satisfies,
\begin{equation*}
M^*(t) = E[N^*(t)] = E[\text{\# of renewals in }[0,t]] = \frac{t}{\mu}.
\end{equation*}~\\
An ordinary renewal process has $X_1 \sim F$ with $F \geq_{st} G$, since $F$ is NBU. It follows that
\begin{align*}
M(t) & =  E[N(t)]\\
     & =  E[\text{\# of renewals in } [0,t] \text{ for ordinary renewal process}]\\
     & \leq  E[N^*(t)] = \frac{t}{\mu},
\end{align*}
and hence
\begin{equation*}
E[M(T_r)] \leq E[\mu^{-1}T_r] = 1.
\end{equation*}
~\\Under the assumption that $T_{kr} - T_{(k-1)r} \geq_{st} T_r$, then
$$
\widetilde{N}_t r \leq \max_{0 \leq s \leq t} X_s \leq (\widetilde{N}_t+1)r,
$$
where $\widetilde{N}_t = \max\{k: T_{kr}\leq t\}$, i.e. the number of renewals prior to time $t$. It follows that
\begin{align}
a(t) & \leq  rE[N_t+1]\nonumber\\
     &   =   (M_t + 1)r\nonumber\\
     & \leq  \left(\frac{t}{\mu}+1\right)r \label{eq:atupper}
\end{align} and thus
\begin{equation*}
E[a(T_r)] \leq \left(\frac{E[T_r]}{\mu}+1\right)r \leq 2r.
\end{equation*}
By plugging in, specifically, $t = E[T_r]$ into \eqref{eq:atupper}, we yield
\begin{equation}
E[T_r] \leq a^{-1}\left(\frac{E[T_r]}{\mu} +1 \right)r = a^{-1}(2r),
\end{equation}
which completes the proof. 
\end{proof}
As mentioned, eq. \eqref{eq:etrupper}, coupled with Eq. (1.1) of \cite{delapenayang}:
\begin{equation}
\label{eq:rightbound}
\frac{1}{2}a^{-1}(\frac{r}{2}) \leq E[T_r] \leq a^{-1}(2r).
\end{equation}
By assuming that $a(t)$ is concave, the lower bound can be improved to $a^{-1}\left(\frac{r}{2}\right)$, which is sharp. If it is further assumed that the conditions specified in Proposition 2.2 hold, we obtain the following bounds:
\begin{equation}
\label{eq:mainresult2}
a^{-1}(\frac{r}{2}) \leq E[T_r] \leq a^{-1}(2r),
\end{equation}
hence gives the right order of magnitude of the expected value of the first hitting time. In fact, \cite{yang} used this approach to obtain bounds on $E[T_r]$ for additive processes including a certain class of stochastic integrals extending the works of Burkh\"older and Gundy; see \cite{burkholdergundy70}.\\


Readers may compare the results obtained in Propositions 2.1 and 2.2 to Theorem 3 of \cite{delapenayang} whose lower bound suggested: $E[T_r] \geq (1-\epsilon) a^{-1}(\epsilon r), 0 < \epsilon < 1$, is sharpened by our new lower bound derived in the case of concave functions. The results as shown in \eqref{eq:rightbound} and \eqref{eq:mainresult2} provide the values of the constants that appear in the bounds as shown in \cite{yang}, which provides examples of the approach applied to stochastic integrals.\\

It can be easily derived from \eqref{eq:eatrlower} and \eqref{eq:upperboundEaTr} that
\begin{equation*}
-\frac{1}{2} \leq \frac{E[a(T_r)]-r}{r} \leq 1,
\end{equation*}
shows stability of $E[a(T_r)]$ as $r$ changes and the linearization property of $a(t)$. In addition, this shows the linearizing property of $a(t)$ as can be easily seen that by ~\\
\begin{equation}
\sup_{r > 0} \bigg|\frac{E[a(T_r)]-r}{r}\bigg| =  \sup_{r > 0} \bigg|\frac{E[\sup_{0 \leq s \leq T_r}\widetilde{X}_s]-E[\sup_{0 \leq s \leq T_r} X_s]}{r}\bigg|\leq 1. \label{eq:rate}
\end{equation}~\\

Equation \eqref{eq:rate} may explain the claim that $a(t)$ can be interpreted as a natural clock for all the processes with the same $a(t)$ since through $a(t)$, $E[a(T_r)]$ is in a linear relationship with the predefined boundary $r$.\\ 

The above framework provides a broad foundation for more general applications. A rich array of examples are given in \cite{delapenayang}. In particular, the above results can be extended easily to $\mathbf{X} \in \mathbb{R}^d$. We will provide the following examples for illustration. 
\begin{ex}
Suppose that we are interested in the first hitting time of the process $X_t$ that hits either a lower bound $a$ or an upper bound $b$, where $a < 0 < b$. We may define
\begin{equation*}
f_{a,b}(x) = \begin{cases} x/a, & x < 0\\x/b, & x\geq 0, \end{cases}
\end{equation*}
(see Figure 1) and hence
\begin{equation*}
T_{a,b} = \inf\{t > 0: X_t \notin [a,b]\} = \inf\{t > 0 : f_{a,b}(X_t)>1\}
\end{equation*}
and
\begin{equation*}
a(t) = E\left[\sup_{0 \leq s \leq t} f_{a,b}(X_s)\right] = E\left[\sup_{0 \leq s \leq t} I(X_s<0)a^{-1}X_s+I(X_s \geq 0)b^{-1}X_s \right].
\end{equation*}

\begin{figure}[!h]
\begin{center}
  \includegraphics[scale=0.5]{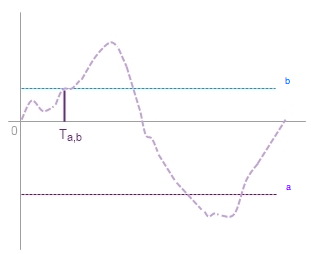}
  \caption{Illustration of asymmetric bounds for general (not necessarily non-negative) processes.} \label{fig:illustration}
\end{center}
\end{figure}
\end{ex}
More generally, this can be further generalized into cases where the first hitting time is defined as $T_r=\inf\{t \geq 0: f_{c_1, c_2, \ldots, c_m}(\boldsymbol X_t) > r\}$. Below shown one specical case that can be handled under this framework: 
One can also take for example, (a) $f(x) = x^+, x \in \mathbb{R}$; (b) $f(x,y) = |x-y|, x, y \in \mathbb{R}^d$; (c) $f(x, y) = \rho(x-y)$ for a metric $\rho$; (d) $f(x,y) = e^{-|x-y|}, x, y \in \mathbb{R}^d$ and so on for appropriate applications. Moreover, the boundary itself can be not fixed. 
\begin{ex}
Suppose that $g$ is a deterministic or stochastic process. Define $Y_t = X_t/g_t$ and $T_r = \inf \{t: Y_t > r\}$, where $X$ is similarly defined as in Example 3. This corresponds to the hitting time of the process $X_t$ reaching $g_t$ and $a(t) = E[\sup_{s\leq t}X_s/g_s]$.
\end{ex}

\section{Some extensions}
Notice that, under the concave assumption on $a(\cdot)$, we can derive a similar lower bound to the one derived in Section 2 that is more readily available. Observe that
\begin{equation*}
a(t) = E\left[\sup_{0 \leq s \leq t}X_s\right] \geq \sup_{0 \leq s \leq t}E[X_s] := \kappa(t).
\end{equation*}
It follows that $a^{-1}(r) \leq \kappa^{-1}(r)$. Thus, if the conditions hold for the upper bound, then 
\begin{equation}
E[T_r] \leq a^{-1}(2r) \leq \kappa^{-1}(2r),
\end{equation}
where $\kappa(\cdot)$ may be obtained more easily, compared to $a(\cdot)$.
\begin{ex}
Consider the absolute value of a standard Brownian motion $|W_t|$. It can be shown that 
\begin{equation*}
\kappa(t) = E|W_t| = \sqrt{2t/\pi}
\end{equation*}
and it is known that $E[\sup_{0 \leq s \leq t}W_s] = E|W_t|$, but $a(t) = E[\sup_{0\leq s\leq t}|W_s|]$ appears difficult to compute. In fact, it equals
\begin{equation*}
c(t)\sqrt{2t/\pi},
\end{equation*}
where $1 < c(t) < 2$, a multiple depending on $t$. It should be noted that the conditions for the upper bound to hold are easily verifable in this case. Thus
\begin{equation*}
E[T_r] \leq \kappa^{-1}(2r) = 2\pi r^2,
\end{equation*}
while it is known that $E[T_r] = r^2$.
\end{ex}
If more information about the process $X_t$ is known, we can obtain a relaxed lower bound that can be easily expressed in a more manageable form. Suppose $X_t$ is a submartingale with right continuous paths, it is well known that
\begin{equation*}
\Pr\left\{\sup_{0 \leq s \leq t}X_s \geq r\right\} \leq \frac{E[X^+_t]}{r} := \frac{\eta(t)}{r}.
\end{equation*}
Therefore, it follows that $\Pr\{T_r \leq t\} = \Pr\{\sup_{0 \leq s \leq t}X_s \geq r\} \leq \eta(t)/r$, which leads to
\begin{align}
\frac{E[\eta(T_r)]}{r} & \geq  \int \Pr\{T_r \leq t\}dF_{T_r}(t) \geq \frac{1}{2}  \nonumber\\
\Longrightarrow \quad\quad\quad E[\eta(T_r)] & \geq  \frac{r}{2}, \label{eq:modre}
\end{align}
which is a stonger result compared with Proposition 2.1.
\begin{ex}
Consider $W^2_t$, where $W_t$ is a standard Brownian motion. Denote $T_r$ the first passage time of $|W_t|$ to $r$, (so for $W^2_t$ the first passage time to $r^2$). Applying the result of \eqref{eq:modre}, we get
\begin{equation*}
E[T_r] \geq \frac{r^2}{2}.
\end{equation*}
Note that the actual value of $E[T_r]$ is $r^2$ in this case. Here, $\eta(t) = E[X^2_t] = t$.
\end{ex}
The usefulness of \eqref{eq:modre} can be demonstrated in the following example in which a closed form expression of $E[T_r]$ is difficult to obtain.
\begin{ex}
Consider a submartingale
\begin{equation*}
X_t = |W_t^2-t|,
\end{equation*}
where $W_t$ is a standard Brownian motion. It can be shown that 
\begin{equation*}
\eta(t) = E|W^2_t - t| = tE|W^2_1 - 1| = \sqrt{\frac{8}{\pi e}}t.
\end{equation*}
Then it follows that
\begin{equation*}
\frac{r}{2} \leq E\left[\eta(T_r)\right] = \sqrt{\frac{8}{\pi e}}E[T_r],
\end{equation*}
so that 
\begin{equation*}
E[T_r] \geq \frac{r}{2}\sqrt{\frac{\pi e}{8}} \approx 0.5116r.
\end{equation*}
In this exmple, the mean of $T_r$, the first passage time of $|W^2_t - t|$ to $r$, is difficult to compute.
\end{ex}
Before ending this section, we would like to point out that the finiteness of the mean of $T_r$ is important because this issue can cause problems in other applications. A good illustration is demonstrated in the following example.
\begin{ex}
Consider
\begin{equation*}
W^+_t = \max\{0, B_t\},
\end{equation*}
where, again, $B_t$ is a standard Brownian motion. The first passage time of $W^+_t$ to $r > 0$ coincides with the first passage time of $W_t$ to $r$ which has an infinite mean. $W^+_t$ is not a Markov process. Hence, results of Proposition 2 do not apply to $W^+_t$ case. $T_r$ is finite with probability 1, but this is not NBU; see definition \eqref{eq:NBU}.\\

Since $W_{s+t} \sim W_s + (W_{s+t}-W_t) = W_s + \widetilde{W}_t$, where $\widetilde{W}_t$ is an independent Brownian motion, 
\begin{equation*}
W^+_{T_r + t} \sim (W_{T_r}+\widetilde{W}_t)^+ = (r + \widetilde{W}_t)^+ \leq r + \widetilde{W}_t^+.
\end{equation*}

When $\widetilde{W}^+_t$ hits $r$, $\widetilde{W}_t^++r$ hits $2r$ but $W^+_{T_r + t}$ may not have hit $2r$. Thus $T_{2r} \geq_{st} T_r + \widetilde{T}_r$ and $T_{2r}-T_r \geq T_r$. More generally, $T_{(k+1)r}-T_{kr} \geq_{st} T_r$. Thus, following our previous argument, we have

\begin{equation*}
rN_t \leq \max_{0 \leq s \leq t}\widetilde{W}_s \leq r(N_t + 1),
\end{equation*}
but in this case, $\{N_t\}$ is a renewal process with an infinite mean interarrival time. There is no stationary distribution $G$ (of $T_r$) and \eqref{eq:atupper} does not hold. In this case, $E[a(T_r)] = E[\sqrt{2/\pi}\sqrt{T_r}] = \infty$. This shows that upper bounding $E[a(T_r)]$ is challenging to deal with in general.
\end{ex}
\section{The rate of growth of the maximum of Bessel processes}
This section is dedicated to the discussion of the rate of growth of the maximum of Bessel processes which can be obtained via the inequalities obtained in Section 2 (and 3). We are going to consider a case in which the hitting time of the radius/surface area/volume of the  largest multi-dimensional Brownian motions that hits a predefined boundary.

\begin{ex}
Let $\mathbf{p_i} = (p^{(1)}_i, p^{(2)}_i, p^{(3)}_i), i = 1, \ldots, m$, be a set of $m$ spheres that perhaps represent some identified tumours in a human body, i.e. a three-dimensional space. At time $t = 0$, we start a three-dimension Brownian motion with coordinates processes centered at each one of these points. For each $i$ a sphere of radius $r_i(t)$ is given, where the radius equals the distance of the location of the three dimensional Brownian motion at time $t$ to its starting point $\mathbf{p}_i$ We are interested in getting qualitative information on how long it will take before the radius (volume) of at least one of the spheres exceeds a fixed level (say $r$, $r>0$) as the size of $m$ varies.

\begin{figure}[!h]
\begin{center}
  \includegraphics[scale=0.5]{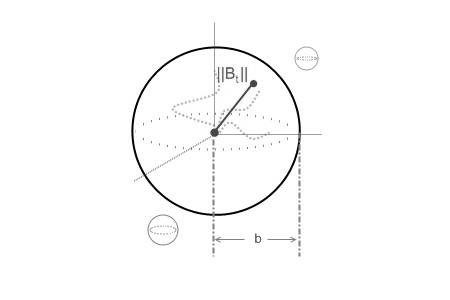}
  \caption{Illustration of Examples 5 and 6.}
\end{center}
\end{figure}

Let $\mathcal{X} = \{\mathbf{X}_i(t) = (X^{(1)}_i(t), X^{(2)}_i(t), X^{(3)}_i(t))\}_{i=1, 2, \ldots, m}$ where, for each $i$, $\mathbf{X}_i(t) = \mathbf{p}_i + \mathbf{B}_i(t)$ corresponds to the location of the three-dimensional Brownian motion $\mathbf{B}_{i}(t)$ started at point $\mathbf{p}_i$ and $\mathcal{B}(t) = \{(B^{(1)}_i(t), B^{(2)}_i(t), B^{(3)}_i(t))\}\newline_{i = 1, \ldots, m}$. The radius of the largest sphere is given by $\|\mathcal{B}(t)\| = \newline\sup_{j\leq m}\sqrt{\sum^3_{i=1}(B^{(i)}_{j}(t))^2}$.\\

Observe that for each $j$, $Y_j(t) = \sum^3_{i=1}(B^{(i)}_j(t))^2$ is a submartingale, since
\begin{equation*}
E\left[Y_j(t) | \mathcal{F}_s\right] = \sum^3_{i=1}E\left[(B_j^{(i)}(t))^2\big|\mathcal{F}_s\right] \geq \sum^3_{i=1}(B^{(i)}_j(s))^2 = Y_j(s).
\end{equation*}
It follows that $Y_t := \max_{i \leq j \leq m}Y_j(t)$ is a submartingale because
\begin{align*}
E[Y_t | \mathcal{F}_s] & =  E[\max\{Y_1(t), \ldots, Y_d(t) | \mathcal{F}_s\}]\\
                       & \geq  \max\{E[Y_j(t) |\mathcal{F}_s], j = 1, \ldots, d\}\\
                       & \geq \max(Y_j(s), j = 1, \ldots, m) = Y_s.
\end{align*}
Let $T_{r,d}$ be the first passage time of $\|B(t)\| = \sup_{j \leq m}\sqrt{\sum^3_{i=1}(B^{(i)}_j(t))^2}$ to $r$, equivalently the first passage time of $Y_t = \|B_t\|^2$ to $r^2$. Then
\begin{align*}
\Pr\{T_{r,d} \leq t\} & = \Pr\left\{\sup_{s \leq t}Y_s \geq r^2\right\}\\
                      & \leq \frac{E[Y_t]}{r^2} = \frac{t}{r^2}E\left[\max_{1 \leq i \leq m}\chi^2_{3,i}\right]. 
\end{align*} 
Therefore, we can write
\begin{equation}
\frac{1}{2} \leq \frac{E\left[\max_{1 \leq i \leq d}\chi^2_{3,i}\right]}{r^2} E[T_{r,d}] \quad \Longrightarrow \quad E[T_{r,d}] \geq \frac{r^2}{2E\left[\max_{1 \leq i \leq d}\chi^2_{3,i}\right]}. \label{eq:lowerbessel}
\end{equation}
It can be shown that $E\left[\max_{1 \leq i \leq d}\chi^2_{3,i}\right] \leq 3\sum^d_{i=1} i^{-1}$ and hence \eqref{eq:lowerbessel} can be rewritten as
\begin{equation}
E[T_{r,d}] \geq \frac{r^2}{6\sum^d_{i=1}i^{-1}},
\end{equation}
which follows from the Corollary on Page 266 of \cite{barlowetall72}. In fact, we can also approximate the values of $E[\max_{1 \leq i \leq d}\chi^2_{3,i}]$ via simulations. The corresponding values for various $d$ are tabulated in Table \ref{table:sim}:

\begin{table}[h]
\begin{center}
\begin{tabular}{ |c|cccccc| }
\hline d & 1 & 2 & 3 & 4 & 5 & 10 \\\hline
$E\left[\sqrt{\max_{1 \leq i \leq d}\chi^2_{3,i}}\right]$ & 1.599 & 1.979 & 2.173 & 2.324 &  2.413 & 2.720 \\\hline\hline
d & 15 & 20 & 30 & 40 & 50 & 100\\\hline
$E\left[\sqrt{\max_{1 \leq i \leq d}\chi^2_{3,i}}\right]$ & 2.875 & 2.987 & 3.132 & 3.237 & 3.307 & 3.527\\\hline
\end{tabular}\\
\begin{tabular}{c}
\\
\end{tabular}
\caption{Approximation of $E\left[\sqrt{\max_{1 \leq i \leq d}\chi^2_{3,i}}\right]$ via 10,000 simulations.}
\label{table:sim}
\end{center}
\end{table}

\end{ex}
In general, for $\{X_i(t), t \geq 0, \}_{i = 1, \ldots, m}$ be iid as $\{X(t), t \geq 0\}$ a submartingale. Let $T^{(r)}_i$ be the first passage time to $r$ for $\{X_i(t), t\geq 0\}$, then $\min_{1 \leq i \leq m}T^{(r)}_i$ is the first passage time to $r$ for $\max_{1 \leq i \leq m}X_i(t)$. Let $\{X(t), t\geq 0\}$ be independent of $\{X_i(t), t \geq 0\}_{i = 1,\ldots, m}$ with $T_r$, its first passage time to $r$. Note that
\begin{equation*}
\int\Pr\{T_r \leq t\}dF_{\min\{T^{(r)}_1, \ldots, T^{(r)}_d\}}(t) = \Pr\left\{T_r \leq \min\{T^{(r)}_1, \ldots, T^{(r)}_d\}\right\} \geq (d+1)^{-1}.
\end{equation*}
Since 
\begin{equation*}
\Pr\{T_r \leq t\} \leq \frac{E[X(t)]}{r} = \frac{\kappa(t)}{r},
\end{equation*}
so 
\begin{align*}
(d+1)^{-1} &\leq  \int\Pr\{T_r \leq t\}dF_{\min\{T^{(r)}_1, \ldots, \min\{T^{(r)}_d\}}(t) \\
           &\leq  \frac{E[\kappa(\min\{T^{(r)}_1, \ldots, T^{(r)}_d\})]}{r}.
\end{align*}
As a result, $E[\kappa(\min\{T^{(r)}_1, \ldots, T^{(r)}_d\})] \geq \frac{r}{d+1}$ For $\|X_t\|$, we use $\|B_t^2\|$ and $r^2$ instead. Recall that $\kappa(t) = E\left[\sum^3_{i=1}(B^{(i)}_j(t))^2\right] = 3t.$ This gives 
\begin{equation*}
E\left[\min_{1 \leq j \leq d}T_i\right] \geq \frac{r^2}{3(d+1)}.
\end{equation*}
This result is not as good as what we can obtain in Example 3.4 in which $E[T_{r,d}] \geq \frac{r^2}{6\sum^d_{i=1}i^{-1}}$. But there we used the result that $\chi^2_3$ is IFR (increasing failure rate); in other examples, we might know the type of distribution that $T^{(r)}_i$ follows.\\

The following example studies the upper bound of the Bessel process discussed in Example 3.4.
\begin{ex}
(Example 5 of a non-Markov process for which the upper bound holds) Consider the Bessel process studied in Example 5 again. $\|B_s\|$ is known to be strongly Markov but $\max_{1\leq j \leq d}\|B^{(j)}_t\|$ may not be. Note that 
\begin{equation*}
B_{t+s} = B_s + (B_{t+s} - B_s) =_d B_s + \widetilde{B}_t,
\end{equation*}
where $\{\widetilde{B}_t, t \geq 0\}$ is distributed as $\{B_t, t \geq 0\}$ and $B$ is independent of $\widetilde{B}$. It follows that
\begin{equation*}
\|B_{t+s}\| =_d \|B_s + \widetilde{B}_t\| \leq_{st} \|B_s\| + \|\widetilde{B}_t\|,
\end{equation*}
and hence 
\begin{equation*}
\|B_{t+s}\| - \|B_s\| \leq_{st} \|\widetilde{B}_t\|, \quad t \geq 0,
\end{equation*}
independently of $\mathcal{F}_s$, the history accumlated up to time $s$. Thus, independently of $\mathcal{F}_s$, the first passage time of $\|B_{t+s}\| - \|B_s\|$ to $r$ is stochastically larger than $T_r$.\\

Now consier $\|\{B^{(j)}\}_t\|$, $j = 1, \ldots, d$. Given $\mathcal{F}_s$, which denotes the history of all $j$ processes up to time $s$, the conditional distribution of each $\|B^{(j)}_{t+s} - B^{(j)}_s\|$ is stochastically larger than each of the $\|\widetilde{B^{(j)}_t}\|$ for all $t \geq 0$. It follows that the first passage time of $Z_t := \max_{1 \leq j \leq d}\|B^{(j)}_t\|$ from $kr$ to $(k+1)r$ is stochastically larger than the minimum of $d$ random variables, each distributed as $T_r$. Hence $T^*_{(k+1)r}- T^*_{kr} \geq_{st} T^*_r$ for all $k \geq 1$. It further follows that, $T^*_{kr}$ is stochastically larger than the convolution of $n$ iid random variables, each distributed as $T^*_r$, where $T^*_r = \min_{1 \leq j \leq d}T_r^{(j)} (= T_{r,d}$).\\

Letting $\widetilde{N}_t = \max\{k: T_{kr}^* \leq t\}, N(t) = \max\{k: \sum^k_{i=1}T^*_{r,i} \leq t\}$, we can write
\begin{eqnarray*}
E\widetilde{N}_t \leq E[N(t)] \leq \left(\frac{t}{E[T_{r,d}]}+1\right)r,
\end{eqnarray*}
since $T^*_r = T_{r,d}$ is NBU. Next, denote $Y_d =_d \max_{1 \leq \nu \leq d}\{\chi^2_{3,\nu}\}$, we have
\begin{equation*}
\sqrt{t}E[\sqrt{Y_d}] = E\left[\max_{1\leq j\leq d}\|B^{(j)}_t\|\right] \leq r\left(\frac{t}{E[T_r^*]}+1\right).
\end{equation*}
It follows that $\sqrt{E[T_r]}E[\sqrt{Y_d}] \leq 2r$ [by letting $t = E[T_{r,d}]$ and $\sqrt{E[T_r]} \leq 2r(E\sqrt{Y_d})$]. As a result, we have
\begin{equation}
E[T_r] \leq \frac{4r^2}{[E\sqrt{Y_d}]^2}. 
\end{equation}
The two-sided bound in this case is thus
\begin{equation*}
\frac{r^2}{4[E\sqrt{Y_d}]^2} \leq E[T_{r,d}] \leq \frac{4r^2}{[E\sqrt{Y_d}]^2},
\end{equation*}
or equivalently,
\begin{equation}
a^{-1}(r/2) \leq E[T_{r,d}] \leq a^{-1}(2r),
\end{equation}
which is the same as what we obtained in eq. \eqref{eq:mainresult2}. Note that, in this example, the underlying process $Z_t = \max_{1 \leq j \leq d}\|B^{(j)}_t\|$ is not Markov. However, each $\|B^{(j)}_t\|$ is Markov and so $T^{(r)}_r$ is NBU. It follows that $T^*_r = \min_{1 \leq j \leq d}T_r^{(j)}$ is NBU and thus has a finite mean.\\

Finally, we would like to emphasize that since the constants involved in the bounds derived are independent of the size of $d$, the inequalities obtained can be used to derive quantitative comparisons on the expected first passage times for processes with different values of $d$. That is, if we include the dependence on $d$ for different values, say $d_1$ and $d_2$, and take ratios, we have 
\begin{equation*}
E[T_{2, r, d_2}^p] \approx \frac{E[\sup_{i \leq d_1}|B_1^{(i)}|^p]}{E[\sup_{i \leq d_2}|B_2^{(i)}|^p]}E[T_{1, r, d_1}^p],
\end{equation*}
which gives us information on the relative growth rate between the maxima of Bessel processes.
\end{ex}
All the examples shown in this section involve Brownian motion, below shown is an example that demonstrate how the bounds can be applied to other types of random variables whose distribution is not Gaussian. 
\begin{ex}
Let $X_1, X_2, \ldots, X_n$ be non-negative, possibly dependent random variables with $\Pr\{X_i = X_j\} = 0$ for all $i \neq j$. Let $T_r$ denote the $r$th smallest amongst $X_1, \ldots, X_n$, $F_i$ be the marginal CDF of $X_i$ and $N(t) = \{\# X_i \leq t\}$. Then, 
\begin{equation*}
a(t) = E[N(t)] = \sum^n_{i=1}F_i(t)
\end{equation*}
and
\begin{equation*}
E[a(T_r)] = E\left[\sum^n_{i=1}F_i(T_r)\right] \geq \frac{r}{2}.
\end{equation*}
Thus, $N(T_r) \equiv r$ and 
\begin{equation*}
E[a(T_r)] = E[\widetilde{N}(T_r)] \geq \frac{r}{2},
\end{equation*}
illustrating the decoupling aspect of this inequality. Again, $\widetilde{N}(t) = \{\# \widetilde{X}_i \leq t \}$, where $\{\widetilde{X}_i\}_{i = 1, \ldots, n}$ are independent copies of $\{X_i\}_{i = 1, \ldots, n}$ It should be noted that this lower bound is obtained without the knowledge of the dependence structure of $X$'s.\\

The above setting can be used to model a pool of debtors whose survival times (time until which they become default) follow some distribution with non-increasing density function, say exponential distribution or a subset of Weibull distribution. In this case, $T_r$ can be interpreted as the time when $r$ (out of $n$) debtors have gone bankrupt, which can be an important time stamp that triggers the termination of payment to a lower tranch of collateralized debt obligation (CDO). By assuming stationarity, we can treat the data as two sets of independent copies and use the historical data to estimate the current set of individuals (random variables). Unlike the use of copula to model the event time, the results presented previously can provide bounds for the event time without knowing the dependence structure of the debtors. 
\end{ex}

\section{Conclusion}
In this paper, we derive bounds for the expectation of the stopping time of arbitrary stochastic processes. The approach we use (see \cite{delapena97} and \cite{delapenayang}) involves the concept of boundary crossing of a non-random function to that of a random function. In the situations where the moment of the maximal process is available, the results shown can be helpful for the estimation of $E[T_r]$. In particular, for non-negative, continuous and time homogenuous Markovian processes, with the assumption that $T_{kr} - T_{(k-1)r} \geq_{st} T_r$ for $k \geq 2$, we show that the order of magnitude is the same up to constants. The result of \eqref{eq:rightbound} suggests that it is appropriate to view the process through $a(t)$, which serves as a clock for all processes with the same $a(t)$, as reflected in \eqref{eq:rate}. The lower bound derived can be applied to arbitray measurable processes and it is particularly useful in the study of renewal processes. \\


\bibliographystyle{spmpsci}

\end{document}